\documentclass[12pt]{article}
\usepackage[utf8]{inputenc}
\usepackage[margin=1in]{geometry}
\usepackage{amsmath}
\usepackage{amssymb}
\usepackage{float}
\usepackage{setspace}
\usepackage{graphicx}
\usepackage{multicol}
\usepackage[normalem]{ulem}
\usepackage{color}
\usepackage{hyperref}
\usepackage{caption}
\usepackage{comment}
\usepackage{multirow}
\usepackage[makeroom]{cancel}
\usepackage{braket}
\usepackage{amsthm}
\usepackage{tikz}
\usepackage{mathtools}
\usepackage{comment}
\usepackage{standalone}
\newcommand{\norm}[1]{\left\lVert#1\right\rVert}
\newcommand{\bm}{\begin{bmatrix}}
\newcommand{\fm}{\end{bmatrix}}
\usepackage{enumitem}
\newtheorem{theorem}{Theorem}
\newtheorem{lemma}{Lemma}
\theoremstyle{definition}
\newtheorem{definition}{Definition}[section]

\title{On the relationships between Z-, C-, and H-local unitaries}

\author{Jeremy Cook}
\date{}
\begin{document}
\maketitle

\begin{abstract}
Quantum walk algorithms can speed up search of physical regions of space in both the discrete-time \cite{ambainis2005discrete-walk-space} and continuous-time setting \cite{childs-2004-spatial}, where the physical region of space being searched is modeled as a connected graph. In such a model, Aaronson and Ambainis \cite{aaronson-spatial-search} provide three different criteria for a unitary matrix to act locally with respect to a graph, called $Z$-local, $C$-local, and $H$-local unitaries, and left the open question of relating these three locality criteria. Using a correspondence between continuous- and discrete-time quantum walks by Childs \cite{childs2010relationship}, we provide a way to approximate $N\times N$ $H$-local unitaries with error $\delta$ using $O(1/\sqrt{\delta},\sqrt{N})$ $C$-local unitaries, where the comma denotes the maximum of the two terms.

\end{abstract}

\section{Introduction}
Spatial search \cite{aaronson-spatial-search}, quantum walks \cite{kempe2003intro}, and quantum cellular automata \cite{watrous1995one} all have the similar problem of defining what it means for a unitary matrix to act `locally' with respect to a graph. For a search algorithm acting on a physical region of space, we must take into account the fact that the speed of light is finite. As in \cite{aaronson-spatial-search}, this suggests that we model our space by a connected graph, where two points in space are connected only if a `quantum robot' (which can occupy a superposition over finitely many points) can travel from one point to the other in one time step. The time needed to search for a marked item would then critically depend on the structure of the graph. This problem is fundamentally related to the locality constraint of quantum walks on graphs, where amplitude can only be transferred between adjacent vertices in each time step. In the case of the continuous-time quantum walk, it's the governing Hamiltonian that must act locally with respect to the graph. In an attempt to relate these locality conditions, Aaronson and Ambainis provide three different locality criteria for unitary operators on graphs called $Z$-local, $C$-local, and $H$-local unitaries, which we will define in Section \ref{definitions}. In Section \ref{relations} we will discuss the known relationships between these locality criteria and methods for relating them which fail. In Section \ref{correspondence} we will review a correspondence between continuous- and discrete-time quantum walks by Childs \cite{childs2010relationship}, and show how the correspondence can be used to relate $H$-local unitaries to $C$-local unitaries.

\section{Locality Criteria}
\label{definitions}
In the classical case, a stochastic matrix $M$ acts locally on a graph $G = (V,E)$ if it does not `send' probability between two vertices which are not connected. More formally, label our vertices $V = \{j\,|\,j = \,1,...,N\}$ with edges $E = \{(j,k)\,|\,j\text{ connected to }k\}$, where we will consider every vertex connected to itself, so that $\forall j$, $(j,j)\in E$. A stochastic matrix $M$ acts locally on a graph if $M_{jk} = 0$ whenever $(j,k)\notin E$. \par
Before defining what it means for a unitary matrix to act locally, we first note a difference between discrete-time quantum walks and continuous-time quantum walks. Continuous-time quantum walks usually act directly on the state space spanned by the vertices of our graph, for example \cite{farhi1998quantum, childs-2004-spatial}. However, applying a discrete-time quantum walk on the state space spanned by the vertices of our graph doesn't always result in interesting walks. For example, let our Hilbert space be $\mathcal{H}_V = \text{span}\{\ket{j}|\,j\in \mathbb{Z}\}$. If we try to define a walk operator which performs
\begin{gather}
    \ket{n} \to a\ket{n-1} + b\ket{n} + c\ket{n+1},
\end{gather}
then it was shown in \cite{ambainis2003quantum, meyer1996quantum} that this operator is unitary if and only if one coefficient is equal to 1 and the other two are equal to 0. To resolve this issue the state space of the discrete-time quantum walk is expanded with a ``coin register'', so that the Hilbert space becomes $\mathcal{H} = \mathcal{H}_V\otimes\mathcal{H}_C$, as shown in \cite{kempe2003intro}. For this fundamental reason, discrete-time quantum walks generally need to act on a larger state space than continuous-time quantum walks. \par
Aaronson and Ambainis \cite{aaronson-spatial-search} define three different criteria for whether a unitary matrix acts locally on a graph. Informally, a unitary matrix is $Z$-local ($Z$ for zero) if it `sends' no amplitude between vertices which are not connected. A unitary matrix is $C$-local ($C$ for composability) if it can be written as a product of commuting unitaries each of which act only on a single edge. A unitary matrix is $H$-local ($H$ for Hamiltonian) if it can be obtained by applying a locally acting, low energy Hamiltonian for some fixed amount of time. By locally acting we mean that the Hamiltonian `sends' no amplitude between vertices which are not connected, i.e. the Hamiltonian is $Z$-local on the graph. \par
Formally, the state of our system at any time can be written as $\sum_{j,z} \alpha_{j,z}\ket{j,z}$ where $j\in V$ and $z$ represents the internal state of our quantum robot, which allows for extra degrees of freedom to account for a coin-register, or additionally for query results in the black-box setting. For all $z$ and $z^*$,

\begin{definition}
$U$ is $Z$-local if $\braket{j,z|U|k,z^*} = 0$ whenever $(j,k) \notin E$.
\end{definition}
\begin{definition}
$U$ is $C$-local if the basis states can be partitioned into subsets $P_1,...,P_q$ such that
\begin{itemize}
    \item[(i)] $\braket{j,z|U|k,z^*} = 0$ whenever $\ket{j,z}$ and $\ket{k,z^*}$ belong to distinct $P_m$'s,
    \item[(ii)] Each $P_m$ either contains basis states the same vertex or from two adjacent vertices.
\end{itemize}
\end{definition}
\begin{definition}
$U$ is $H$-$local$ if $U = e^{-iHt}$ for some fixed time $t$ and for some Hermitian matrix $H$ with eigenvalues of absolute value at most $\pi$, such that $\braket{j,z|H|k,z^*} = 0$ whenever $(j,k)\notin E$.
\end{definition}

Discrete-time quantum walks are written as the product of two $C$-local operations, the coin operation followed by the step operation, as in \cite{kempe2003intro, childs-2004-spatial}. Continuous-time quantum walks can be implemented with $H$-local operations. Note however that $H$-local operations do not encompass all continuous-time quantum walks because we require the eigenvalues of our governing Hamiltonian to be bounded. 

\section{Relationships between Z-, C-, and H-local unitaries}\label{relations}
\subsection{Known relationships}
\begin{theorem}
\normalfont \cite{aaronson-spatial-search} \itshape
Every $C$-$local$ unitary is $Z$-$local$.
\end{theorem}
\begin{proof}
If the vertices $j$ and $k$ are not connected, then their basis states belong in different partitions. Therefore $\forall z,z^*$, $\braket{j,z|U|k,z^*} = 0$.
\end{proof}
\begin{theorem}
\normalfont \cite{aaronson-spatial-search} \itshape
Every $C$-$local$ unitary is $H$-$local$.
\end{theorem}
\begin{proof}
By the spectral theorem, every unitary matrix $U$ can be written as $U=e^{-iH}$ for some Hermitian matrix $H$ with eigenvalues of absolute value at most $\pi$. $U$ is normal, so it can be diagonalized as $U = VDV^{\dagger}$. The eigenvalues of a unitary matrix have absolute value at most 1, so we can find real $\theta_k$ (with $|\theta_k| \leq \pi$) such that $e^{-i\theta_k} = D_{kk}$. Define the diagonal matrix $\Lambda_{kk} = \theta_k$, then $U = Ve^{-i\Lambda}V^{\dagger} = e^{-iV\Lambda V^\dagger}$. $\Lambda$ is a real matrix, so $V\Lambda V^\dagger$ is some Hermitian matrix $H$. So we can write each unitary $U_m$ acting on the basis states in $P_m$ as $e^{-iH_m}$. All the $U_m$ commute, therefore
\begin{gather}
    \prod_{m=1}^q U_m = e^{-i\sum_{m=1}^q H_m}.
\end{gather}
\end{proof}

\subsection{Attempted relationships}
To relate $H$-local and $Z$-local unitaries we might wonder if we could modify $U = e^{-iHt}$ in some way such that $U$ would be $Z$-local. We could use the Lieb-Robinson bounds to bound the terms in the matrix $\braket{j|e^{-iHt}|k}$ for $(j,k)\notin E$, corresponding to vertices which are `far apart' in the sense that no signal should be able to propagate between the two vertices in one time step. The Lieb-Robinson bounds state that the amount of correlation between two observables on disjoint regions of space decays exponentially in distance, so we could then find some time $t$ for every $\epsilon > 0$ such that the corresponding matrix elements $\lvert\braket{j|e^{-iHt}|k}\rvert \leq \epsilon$. After bounding these terms, we would define the new matrix $U'$ which is equal to $U$ but with the elements $U'_{jk}$ set to 0 for $(j,k)\notin E$. Then we could apply the Gram-Schmidt orthogonalization procedure to make our matrix unitary again. However, in general this procedure fails because the constraints of unitarity on $U'$ are too strong. Therefore without a canonical way of extending our state space this procedure fails to approximate $H$-local unitaries well with $Z$-local unitaries.\par
Can we use trotterization to relate $H$-local to $Z$-local or $C$-local unitaries? Suppose our Hamiltonian can be written as the sum of local Hamiltonians $H_m$, where each $H_m$ acts on the subset of vertices $n_m$. We can find a sequence of unitary operations which give an approximation for the $H$-local unitary by trotterization,
\begin{gather*}
    e^{-it\sum_{m}H_m} \approx \left(\prod_m e^{-itH_m/n}\right)^n.
\end{gather*}
If each Hamiltonian $H_m$ acts only on the basis states $n_m$, then each unitary operator $e^{-itH_m}$ also acts only on the vertices $n_m$. However, $e^{-itH_m}$  does not necessarily satisfy the locality constraint imposed by the underlying graph. Therefore it is not generally true that $e^{-itH_m}$ is $Z$-$local$ on the set of $n_m$ qubits.\par
In terms of relating $Z$-local to $C$-local or $Z$-local to $H$-local unitaries I was not able to derive any nontrivial results.

\section{Correspondence between continuous- and discrete-time quantum walks}
\label{correspondence}
In order to analyze the relationships between $H$-$local$ and $C$-$local$ operations we use a correspondence between continuous- and discrete-time quantum walks derived by Childs \cite{childs2010relationship}. The correspondence goes as follows. \par
Let $H$ be a Hermitian matrix acting on the orthonormal basis $\{\ket{j}\, |\, j=1,...,N\}$ spanned by our vertices, where $H$ is $Z$-$local$ on the graph. Let $\text{abs}(H) = \sum_{j,k}\lvert H_{jk}\rvert\ket{j}\bra{k}$, and let $\ket{d} = \sum_{j=1}^N d_j\ket{j}$ be the principal eigenvector $\text{abs}(H)\ket{d} = \norm{\text{abs}(H)}\ket{d}$, where without loss of generality we can assume all the coefficients $d_j$ are positive. Let $S$ be the swap operation $S\ket{j,k} = \ket{k,j}$ and define a set of $N$ quantum states $\{\ket{\psi_j}\in\mathbb{C}^N\times\mathbb{C}^N\,|\, j = 1,...,N\}$:
\begin{gather}
    \ket{\psi_j} = \frac{1}{\sqrt{\norm{\text{abs}(H)}}}\sum_{j=1}^N \sqrt{H_{jk}^*\frac{d_k}{d_j}}\ket{j,k}.
\end{gather}
A discrete-time quantum walk corresponding to $H$ is obtained by reflecting about the span of $\{\ket{\psi_j}\}$ and applying $S$. Let
\begin{gather}
    T = \sum_{j=1}^N \ket{\psi_j}\ket{j},
\end{gather}
be the isometry mapping $\ket{j}\in\mathbb{C}^N$ to $\ket{\psi_j}\in\mathbb{C}^N\times\mathbb{C}^N$ so that $TT^\dagger$ is the projector onto $\text{span}\{\ket{\psi_j}\}$, and $2TT^\dagger - 1$ is a reflection about the $\text{span}\{\ket{\psi_j}\}$. To simulate the continuous-time walk by the discrete-time walk given an initial state $\ket{\varphi} \in \mathbb{C}^N$:
\begin{itemize}
    \item[1.] Apply the isometry $T$ followed by the unitary operator $\frac{1+iS}{\sqrt{2}}$.
    \item[2.] Apply $\tau$ steps of the discrete quantum walk operator: $U = iS(2TT^\dagger - 1)$.
    \item[3.] Project onto the basis states $\{\frac{1+iS}{\sqrt{2}} T\ket{j}\,|\,j=1,...,N\}$.
\end{itemize}
Letting $h=\norm{H}/\norm{\text{abs}(H)}$, Childs showed \cite[Theorem 2]{childs2010relationship} that for $\tau=t\norm{\text{abs}(H)}$, the discrete-time quantum walk described above satisfies
\begin{gather}
    \norm{T^\dagger \frac{1-iS}{\sqrt{2}}(U)^\tau \frac{1+iS}{\sqrt{2}}T - e^{-iHt}} \leq h^2\left(1 + \left(\frac{\pi}{2}-1\right)h\tau\right).
    \label{eqn:approx}
\end{gather}
\par \ 
\begin{theorem}\label{thm:q}
The discrete-time quantum walk described above can be modified to be implemented with $O(\tau)$ $C$-local operations.
\end{theorem}
First we note that the swap operation only acts locally on complete graphs, for if $(j,k)\notin E$, then $S\ket{j,k} = \ket{k,j}$ transfers amplitude between two vertices which are not connected. Therefore $S$ is not $C$-local or even $Z$-local. However, we can replace $S$ by the conditional swap operator, where basis states are only swapped if they share an edge in the graph:
\begin{gather}
    Q\ket{j,k} = \begin{cases}
        \ket{j,k} & \text{if } \text{ and } (j,k)\notin E,\\
        \ket{k,j} & \text{otherwise}.
    \end{cases}
\end{gather}

The modified discrete-time walk can be broken down as follows:
\begin{gather}
    \underbrace{\vphantom{\frac{1-iQ}{\sqrt{2}}} T^\dagger}_{a} \bigg(\underbrace{\frac{1-iQ}{\sqrt{2}}}_{b}\bigg) ( \underbrace{\vphantom{\frac{1-iQ}{\sqrt{2}}} iQ}_{c} (\underbrace{\vphantom{\frac{1-iQ}{\sqrt{2}}} 2TT^\dagger - 1}_{d}) )^\tau \bigg(\underbrace{\frac{1+iQ}{\sqrt{2}}}_{e} \bigg) \underbrace{\vphantom{\frac{1-iQ}{\sqrt{2}}} T}_{f}.
    \label{eqn:q}
\end{gather}

The operations $a,b,c,d,e,f$ are $C$-local operations and for even powers of $\tau$ the product $abcdef$ is equal to the walk described in (\ref{eqn:approx}). For completeness we provide these proofs in the Appendix.\par
In order to make the approximation in (\ref{eqn:approx}) arbitrarily good, Childs implements a \textit{lazy quantum walk} in order to make $h = \norm{H}/\norm{\text{abs}(H)}$ arbitrarily small. To do this, construct a set of $N$ orthogonal states $\{\ket{\bot_j}\,|\,j=1,...,N\}$ that satisfy
\begin{gather}
    \braket{\psi_j|\bot_k} = \braket{\psi_j|S|\bot_k} = \braket{\bot_j|S|\bot_k} = 0\text{ for } j,k = 1,...,N,
\end{gather}
modify the states $\ket{\psi_j}$ to
\begin{gather}
    \ket{\psi_j^\epsilon} = \sqrt{\epsilon}\ket{\psi_j} + \sqrt{1 - \epsilon}\ket{\bot_j}.
\end{gather}
and modify the isometry $T$ to $T_\epsilon = \sum_j\ket{\psi_j^\epsilon}\bra{j}$. Using the fact that $\braket{\psi_j^\epsilon|S|\psi_k^\epsilon} = \epsilon\braket{\psi_j|S|\psi_k}$, it was shown by Childs that $h$ in (\ref{eqn:approx}) is replaced by $\epsilon h$. Choosing $\epsilon = \norm{\text{abs}(H)}t/\tau$, the accuracy of the simulation can be increased by increasing the number of steps in the discrete-time walk. To achieve error $\delta$ for any initial state $\ket{\varphi}\in\mathbb{C}^N$,
\begin{gather}
    \norm{\left(T^\dagger \frac{1-iS}{\sqrt{2}}(U)^\tau \frac{1+iS}{\sqrt{2}}T - e^{-iHt}\right)\ket{\varphi}} \leq \delta,
\end{gather}
the complexity of the discrete-time quantum walk must be
\begin{gather}
    \tau \geq \max \left\{\norm{H}t\sqrt{\frac{1+(\frac{\pi}{2}-1)\norm{H}t}{\delta}},\, \norm{\text{abs}(H)}t\right\}.
    \label{eqn:cplx}
\end{gather}
In other words, the complexity is $O((\norm{H}t)^{3/2}/\sqrt{\delta}, \norm{\text{abs}(H)}t)$, where the comma denotes the maximum of the two terms.
\par \
\begin{theorem}\label{thm:c}
Every $N\times N$ $H$-local unitary acting on a graph with more edges than vertices can be approximated with error $\delta$ using $O(1/\sqrt{\delta},\sqrt{N})$ $C$-local unitaries.
\end{theorem}

\begin{proof}
In order to obtain an arbitrarily good approximation without violating locality we must now modify the \textit{lazy quantum walk} provided by Childs. Specifically, to construct the states $\ket{\bot_j}$ Childs enlarged the Hilbert space from $\mathbb{C}^N\times\mathbb{C}^N$ to $\mathbb{C}^{N+1}\times\mathbb{C}^{N+1}$ and chose $\ket{\bot_j} = \ket{j,N+1}$ for simplicity. In terms of our graph structure and the swap operator, this would correspond to adding a vertex to our graph which is connected to every other vertex, which obviously violates locality in our model. Instead, if the number of edges in our graph is strictly greater than the number of vertices (i.e. our graph is not a tree), then we can construct a lazy walk which is local by extending our Hilbert space to $\mathbb{C}^N \times \mathbb{C}^N \times \mathbb{C}^2$. Redefine the states $\ket{\psi_j}$ to be
\begin{gather}
    \ket{\psi_j'} = \frac{1}{\sqrt{\norm{\text{abs}(H)}}}\sum_{j=1}^N \sqrt{H_{jk}^*\frac{d_k}{d_j}}\ket{j,k,0}.
\end{gather}

Let our set of orthogonal states be $\ket{\bot_j'} = \ket{j,k,1}$, where for each vertex $j$, we pick a connected vertex $k$ such that no edge $(j,k)$ is picked twice. With $S'$ swapping the first two registers,
\begin{gather}
    \braket{\psi_j'|\bot_k'} = \braket{\psi_j'|S'|\bot_k'} = \braket{\bot_j'|S'|\bot_k'} = 0\text{ for } j,k = 1,...,N.
\end{gather}
Then the fact $\braket{\psi_j'^\epsilon|S'|\psi_k'^\epsilon} = \epsilon\braket{\psi_j'|S'|\psi_k'}$ allows $h$ in (\ref{eqn:approx}) to go to $\epsilon h$ without violating locality. Replace $\ket{\psi_j^\epsilon}$ by
\begin{gather}
    \ket{\psi_j'^\epsilon} = \sqrt{\epsilon}\ket{\psi_j'} + \sqrt{1 - \epsilon}\ket{\bot_j'}.
\end{gather}
and $T_\epsilon$ by $T_\epsilon' = \sum_j\ket{\psi_j'^\epsilon}\bra{j}$. All these states only contain basis states $\ket{j,k,*}$ with $(j,k)\in E$ so $Q$ has the same action as $S$, and the discrete-time quantum walk remains $C$-local.

Recall that $H$ is Hermitian with bounded eigenvalues, so $\norm{H}\leq \pi$. For a matrix $H \in \mathbb{C}^{N\times N}$, the best possible bounds \cite{absh} for $\norm{\text{abs}(H)}$ is
\begin{gather}
    \norm{H} \leq \norm{\text{abs}(H)} \leq \sqrt{N}\norm{H}.
\end{gather}
Then since $e^{-iHt}$ is applied for a fixed amount of time, the complexity of the walk given by (\ref{eqn:cplx}) reduces to that of Theorem \ref{thm:c}.
\end{proof}

In some cases, (\ref{eqn:cplx}) may be overly pessimistic \cite{childs2010relationship}, and one might wonder if the $\sqrt{N}$ in Theorem \ref{thm:c} could be reduced to polylog($N$). However, there exists a family of Hamiltonians with $\norm{\text{abs}(H)} \gg \norm{H}$ which cannot be simulated in time $\text{poly}(\norm{H}t, \log N)$ \cite{kothari2010}, so the $\sqrt{N}$ cannot be reduced in general.
\par
In summary, this result provides a method for converting a continuous-time quantum walk into a discrete-time quantum walk without violating locality. Alternatively, it can be used as a method for simulating a locally acting Hamiltonian using only local unitary operations. The compilation of relationships between $Z$-local, $C$-local, and $H$-local unitaries are shown in the Figure \ref{fig:tri}, where $*$ means by approximation. 
\begin{figure}[H]
    \centering
    \tikzset{every picture/.style={line width=0.75pt}} 

\begin{tikzpicture}[x=0.75pt,y=0.75pt,yscale=-1,xscale=1]

\draw (73.41,16.39) node [scale=1] [align=left] {C};
\draw (13.17,124.97) node [scale=1.2] [align=left] {Z};
\draw (133.65,124.97) node [scale=1.2] [align=left] {H};
\draw (119,63) node   {$*$};
\draw (74.41,136) node   {$*$};
\draw    (67.31,27.39) -- (21.35,110.22) ;
\draw [shift={(20.38,111.97)}, rotate = 299.02] [color={rgb, 255:red, 0; green, 0; blue, 0 }  ][line width=0.75]    (10.93,-3.29) .. controls (6.95,-1.4) and (3.31,-0.3) .. (0,0) .. controls (3.31,0.3) and (6.95,1.4) .. (10.93,3.29)   ;

\draw    (123.15,124.97) -- (24.67,124.97) ;
\draw [shift={(22.67,124.97)}, rotate = 360] [color={rgb, 255:red, 0; green, 0; blue, 0 }  ][line width=0.75]    (10.93,-3.29) .. controls (6.95,-1.4) and (3.31,-0.3) .. (0,0) .. controls (3.31,0.3) and (6.95,1.4) .. (10.93,3.29)   ;

\draw    (83.88,24.95) -- (132.16,111.97) ;

\draw [shift={(82.91,23.2)}, rotate = 60.98] [color={rgb, 255:red, 0; green, 0; blue, 0 }  ][line width=0.75]    (10.93,-3.29) .. controls (6.95,-1.4) and (3.31,-0.3) .. (0,0) .. controls (3.31,0.3) and (6.95,1.4) .. (10.93,3.29)   ;
\draw    (122.18,114.6) -- (73.8,27.39) ;

\draw [shift={(123.15,116.35)}, rotate = 240.98] [color={rgb, 255:red, 0; green, 0; blue, 0 }  ][line width=0.75]    (10.93,-3.29) .. controls (6.95,-1.4) and (3.31,-0.3) .. (0,0) .. controls (3.31,0.3) and (6.95,1.4) .. (10.93,3.29)   ;

\end{tikzpicture}
    \caption{Relationships between $Z$-local, $C$-local, and $H$-local criteria.}
    \label{fig:tri}
\end{figure}
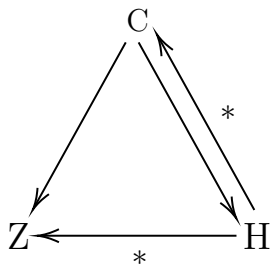

\section{Acknowledgements}
I would like to thank Dr. Scott Aaronson and Patrick Rall for useful discussions and ideas on relating these three locality criteria.

\bibliographystyle{plain}
\bibliography{sources}

\section{Appendix}\label{Appendix}
\subsection{Proof of Theorem \ref{thm:q}}
\begin{lemma}
Let $V = iQ(2TT^\dagger - 1)$ and $W = iS(2TT^\dagger - 1)$. For even integers of $\tau$, $V^{2\tau} = W^{2\tau}$.
\end{lemma}
\begin{proof}
Equivalently we can show that $\forall j,k\in \{1,...,N\}$, $(V^2 - W^2)\ket{j,k} = 0$, using the fact that if $V^2 = W^2$, then $(V^2)^\tau = (W^2)^\tau$. \par
\ \\
\indent\textbf{Case 1:} $(j,k) \notin E$. In this case $S$ and $Q$ have a different action on the basis states $\ket{j,k}$, but $S^2\ket{j,k} = Q^2\ket{j,k}$. Applying one iteration of the reflection about $\text{span}\{\ket{\psi_j}\}$,
\begin{gather}
    (2TT^\dagger - 1)\ket{j,k} = 2\sum_{l=1}^N\ket{\psi_l}\braket{\psi_l|j,k} - \ket{j,k} \\ = 2\sqrt{\frac{H_{jk}\frac{d_j}{d_k}}{\norm{\text{abs}(H)}}}\ket{\psi_j} - \ket{j,k} = -\ket{j,k}
    \label{eqn:mid2}
\end{gather}
The last equality is due to the fact that our governing Hamiltonian is $Z$-local, so if $(j,k)\notin E$ then $H_{jk} = 0$. For notational convenience let $A = 2TT^\dagger - 1$. For $V^2$,
\begin{gather*}
    V^2\ket{j,k} = (i^2)QAQA\ket{j,k} = (i^2)(-1)QAQ\ket{j,k} = (i^2)(-1)QA\ket{j,k} \\= (i^2)Q\ket{j,k} = (i^2)\ket{j,k}.
\end{gather*}
For $W^2$,
\begin{gather*}
    W^2\ket{j,k} = (i^2)SASA\ket{j,k} = (i^2)(-1)SAS\ket{j,k} = (i^2)(-1)SA\ket{k,j} \\= (i^2)S\ket{k,j} = (i^2)\ket{j,k}.
\end{gather*}
Therefore in the case $(j,k)\notin E$, $(V^2 - W^2)\ket{j,k} = 0$.
\ \\ \ \\
\indent\textbf{Case 2:} $(j,k) \in E$. Continuing from (\ref{eqn:mid2}), except now $H_{jk} \neq 0$. Expanding the state $\ket{\psi_j}$, we note that $\braket{j,k|\psi_j} = 0$ whenever $(j,k)\notin E$. Therefore our state is a sum over basis states $\ket{j,k}$ for which $(j,k)\in E$, so the application of $Q$ is identical to $S$ on this state. The second iteration of the walk operator acts in the same manner, therefore $(V^2 - W^2)\ket{j,k} = 0$.
\end{proof}

\begin{lemma}
$\forall j\in \{1,...,N\}$, $QT\ket{j} = ST\ket{j}$.
\end{lemma}
\begin{proof}
$T\ket{j} = \ket{\psi_j}$, and if $(j,k)\notin E$ then $H_{jk}$ is 0, so $\ket{\psi_j}$ contains no states $\ket{j,k}$ for which $(j,k)\notin E$. Therefore $Q$ and $S$ have the same action on these states.
\end{proof}

\begin{lemma}
$TT^\dagger$ is $C$-local.
\end{lemma}
\begin{proof}
$\forall j,j^*,k,k^* \in \{1,...,N\}$ where $j\neq j^*$,
\begin{gather}
    \braket{j,k|TT^\dagger|j^*,k^*} = \sum_{l=1}^N \braket{j,k|\psi_l}\braket{\psi_l|j^*,k^*} \\
    =  \frac{1}{\norm{\text{abs}(H)}} \sum_{l=1}^N \sqrt{H_{jk}^* H_{j^* k^*}\frac{d_k d_{k^*}}{d_j d_{j^*}}}\delta_{jl}\delta_{j^*l} = 0.
    \label{maineq}
\end{gather}
Therefore the operation $TT^\dagger$ is $C$-$local$ because we can partition the basis states into sets $P_j = \{\ket{j,k}\, |\, \forall k \in \{1,...,N\}\}$. Furthermore multiplying by 2 and subtracting the identity does not change this, so $2TT^\dagger - 1$ is $C$-$local$.
\end{proof}
The operation $Q$ is $C$-local by definition, and $T$ is also $C$-local in the sense that for $(j,j^*)\notin E$, $\braket{j^*,k|T|j} = \braket{j^*,k|\psi_j} = 0$. So for even powers of $\tau$ the two walks are equivalent:

\begin{gather}
    T^\dagger \frac{1-iS}{\sqrt{2}}(iS(2TT^\dagger - 1))^\tau \frac{1+iS}{\sqrt{2}}T = T^\dagger \frac{1-iQ}{\sqrt{2}}(iQ(2TT^\dagger - 1))^\tau \frac{1+iQ}{\sqrt{2}}T.
\end{gather}
\end{document}